\documentclass[draftcls,onecolumn,conference,12pt]{IEEEtran}

\usepackage{etoolbox}
\newtoggle{onecolumnformat}
\togglefalse{onecolumnformat}
\newcommand{\vect}[1]{{\mathbf{#1}}}

\newcommand\suchthat{\mid}

\def\booknames#1/#2/#3/#4{#1:#2:#3:#4}
\usepackage{amsmath}
\allowdisplaybreaks
\usepackage{amsthm}
\usepackage[nothing]{algorithm}
\usepackage{algorithmic}
\usepackage{stfloats}
\usepackage{graphicx}
\usepackage{epsfig}
\graphicspath{ {/}{figures/} }

\usepackage{dsfont}
\usepackage{graphicx}
\usepackage{epsfig}
\usepackage{bbm}
\usepackage{dsfont}
\usepackage{amsfonts}
\newcounter{Problemset}
\usepackage{hyperref}

\usepackage{svn}
\newtheorem{lem}{Lemma}

\title{Fast Heuristics for Power Allocation in Zero-Forcing OFDMA-SDMA  Systems with Minimum
Rate Constraints}
\begin{document}
\maketitle
%
%
%
%
%
\begin{abstract}
  We investigate in this paper the optimal power allocation in an OFDM-SDMA system
  when some users have minimum downlink transmission rate requirements. We
  first solve the unconstrained power allocation problem for which
  we propose a fast zero-finding
  technique that is guaranteed to find an optimal solution, and an
  approximate solution that has lower complexity but is not guaranteed to
  converge.
For the more complex minimum rate constrained problem, we propose two approximate
  algorithms. One is an iterative technique that finds an optimal
  solution on the rate boundaries so that the solution is feasible,
  but not necessarily optimal. The other is not iterative but cannot
  guarantee a feasible solution.
  We present numerical results showing that the computation time for
  the iterative heuristic is one order of magnitude faster than finding the exact solution with a numerical solver, and the non-iterative technique is an additional order of magnitude faster than the iterative heuristic. We also show that
  in most cases, the amount of infeasibility with the non-iterative technique is
  small enough that it could probably be used in practice.
\end{abstract}
\section{Introduction}
Due to the increasing bandwidth requirements of wireless users,
new sophisticated spectrum management techniques are required.
An approach that has  recently attracted a lot of interest to increase
throughput  is to
exploit the spatial, frequency  and multi-user diversity
dimensions by combining the orthogonal frequency division
multiplexing access (OFDMA) and spatial division multiple access
(SDMA) techniques.

OFDMA provides multi-user frequency diversity
by dividing the available bandwidth into independent subchannels
and then using a channel-aware scheduler to grant access to the
users with the best conditions for each subchannel. Meanwhile, the
SDMA technique assigns the same frequency
subchannel to a group of users with compatible channel vectors,
thus increasing the system's spectral efficiency. Zero forcing
(ZF) is a practical SDMA beamforming i.e., precoding technique
where  the beamforming vectors are simply computed using the
pseudo-inverse matrix~\cite{wiesel07}.   This cancels inter-user
interference and allows simultaneous parallel transmission
to the selected users. The ZF precoding technique
is not optimal, but its implementation in practical systems is
much simpler than other precoding techniques and is quasi-optimal
at higher signal to noise ratio (SNR) or when users have
quasi-orthogonal channel vectors~\cite{yoo06}.

The constrained
resource allocation (RA) problem for a system using the ZF
OFDMA-SDMA technique is made up of two parts.
First, select   optimally   which users to
assign to each subcarrier, and choose their transmit power allocation to
maximize the sum rate subject to a total power budget and minimum
rate requirements for a subset of users.

The first part of the RA problem --- user selection ---  can
be solved using heuristic methods that scan the users' spatial
vectors and pick the best users for each subchannel.
The second part --- power allocation (PA) ---
can be transformed into a one-dimensional root finding problem when formulated
in the dual domain for the case of one total power constraint and
yields the well-known water-filling power allocation
scheme~\cite{wang-11}. Heuristics
are proposed in~\cite{palomar-03, ling12, jang-03},
to obtain the optimal water level efficiently for the case of
single input single output (SISO) wireless link. For the case of a  multiple input
single output (MISO) wireless link, such as for multi-antenna SDMA transmission, the combined heuristics for subchannel
user selection and power allocation give results close to the
optimum~\cite{dimic05,perea10}.

When the base station (BS) supports best effort (BE) and real-time
(RT) traffic simultaneously, the RA algorithm should guarantee a
certain minimum rate for the RT users, while the BE users with
good channel conditions should be assigned resources to increase
the sum rate.  In this case, the user selection should pick the
subchannels to insure that the RT users get their minimum rates.
Heuristic methods have been proposed
in~\cite{papoutsis11,perea-heur13} for user selection. Efficient
power allocation is also a central tool to improve the performance
of these heuristics~\cite{perea-heur13}.  After an initial user
selection, PA is performed and the user rates are computed. If these
rates are lower than the rate requirements for the RT users, a new
user selection is done and PA is performed again until a good
enough solution is found.

For a given user selection, the resulting PA problem is
convex due to the use of the pseudo-inverse technique. It  can be
solved optimally using a number of standard optimization
techniques~\cite{bertsekas03}.  This can take some time when
the number of subchannels is large, such as in LTE-Advanced systems~\cite{lte-rel10}. In
addition, the PA problem with fixed subchannel assignment is
solved many times by the subchannel assignment
heuristics~\cite{papoutsis11, perea-heur13} when looking for additional
subchannels  to satisfy the RT constraints. For this reasons, we need efficient heuristics that solves the PA
problem for fixed subchannel assignment much faster than the
standard algorithms and with an achieved rate not too different from
optimal solution.  The objective of this paper is the design and
evaluation of such a PA heuristic for ZF OFDMA-SDMA systems with RT users with minimum rate requirements.

Constant power distribution among users has been used as a simple
heuristic for PA. In a scenario where the minimum rate
constraints are lower than the rates achieved by the unconstrained
sum rate maximization solution, the performance obtained is close
to the optimal for BE traffic~\cite{jang03}. However, when the
minimum rates constraints are active, the method cannot always provide
feasible points. An adaptable PA scheme should decrease the power of the
BE users and give it to the RT users to support their minimum
rates. Obviously, this is not possible in a constant
power allocation scheme. This is partially solved in~\cite{anas04},
by assigning constant power to the BE users and performing
water-filling PA on the RT users. In this paper, we go several
steps beyond this and use multi-level water-filling PA over all
users and an efficient heuristic that operates in the dual
domain.

Because we use the pseudo-inverse approximation, the resulting PA
problem is convex. If there are no rate constraints, it can be
solved exactly and efficiently with the methods proposed
in~\cite{ling12,yu-boyd-04,jang-03}. The fastest is that
of~\cite{ling12}  by one  order of magnitude. Several heuristic
methods have been proposed to solve the power minimization problem
under rate constraints for SISO systems \cite{zhu09, kivanc03}.
But, to the best of our knowledge, fast methods to solve the
problem under minimum rate constraints have not been reported for
ZF OFDMA-SDMA systems.

One common method to solve convex problems is dual Lagrange
decomposition~\cite{boyd04}. We define dual
variables for the power and rate constraints. For problems
like ours,  we can get a closed-form expression of the dual
function and the users' powers. The dual function is
continuous but not differentiable and we could use the  subgradient algorithm
to get the optimal dual variables.  The  dual variable for the power constraint
indicates the water level and the ones for the rate constraints
the superimposed multi-levels~\cite{tao-08}. The
size of the problem
grows with the number of RT users, since each rate constraint is
assigned a dual variable, which can lead to large computation times. Instead
of solving the dual problem to
optimality, our method finds an approximate solution that
generates a feasible point to the problem which is much faster to
compute.

The main contribution of this paper is a
method to find the dual variables that satisfy the power and rate
constraints much faster than solving the dual problem optimally
using iterative methods. As our numerical results indicate, the
resulting dual variables generate a feasible point that is close
to the optimal solution and the computation times are several
orders of magnitude lower than for the optimal method.

In section~\ref{sec:prob-form}, we describe the system under
consideration and mathematically formulate the PA optimization
problem. Then we present in section~\ref{subsec:max_thr_heur} some
algorithms for the unconstrained maximum throughput power
allocation. Next, we show  in section~\ref{subsec:rat_con_heur}
how to compute fast approximations when rate constraints are
present. We then  numerically evaluate the proposed algorithm in
section~\ref{sec:results}, both in terms of accuracy and CPU time.
Finally, we present our conclusions in
section~\ref{sec:conclusions}.
\section{Problem formulation and optimal solution}
\label{sec:prob-form}
We consider the resource allocation problem for the downlink
transmission in a multi-carrier multi-user MISO system with a
single BS. There are $K$ users, some of which have RT traffic with
minimum rate requirements, while the others have non real-time
(nRT) traffic that can be served on a best-effort basis. The BS is
equipped with $M$ transmit antennas and each user has one receive
antenna. The system's available bandwidth $W$ is divided into $N$
subchannels whose coherence bandwidth is assumed larger than
$W/N$, thus each subchannel experiences flat fading. In the system
under consideration the BS transmits data in the downlink
direction to different users on each subchannel by performing
linear beamforming precoding. At each OFDM symbol, the BS changes
the beamforming vector for each user on each subchannel to
maximize a weighted sum rate.

To simplify the RA problem, we assume that the beamforming vectors
are chosen according to the ZF criteria, which is known to be
nearly optimal when the SNR is high~\cite{yoo06}. We assume that
we have already chosen a set of users $S_n$ out of the total
number of users $K$ for each subchannel $n$, such that $g_n \doteq
|S_n| \leq M$. Let us define the subcarrier channel matrix
$\vect{H}_n = [ \vect{h}_{ n,s_n(1)}^T \dots \vect{h}_{ n,
s_n(g_n) }^T ]^T$ where $\vect{h}_{n,k} \in \mathbb{C}^{1 \times M
}$ is the channel row $M$-vector between the BS and user $k$ on
subchannel $n$, and
\begin{align}
\label{eq:mat_gamma_def}
 \beta_{n,k} \doteq &
\begin{cases}
{\left[ (\vect{H}_n^{\dag})^{H} \vect{H}_n^{\dag} \right]}_{j,j}
& \text{if } k= s_n(j), \quad \forall j \in \{1,\ldots,g_n\}
\\
0,  & \text{otherwise.}
\end{cases}
\end{align}

Our objective
is to find a fast method to assign the transmit power to users.
First, let us define the problem parameters
\begin{description}
\item[$\check{P}$] Total power available;
\item[$c_k$] Scheduling weight of user $k$;
\item[$\beta_{n,k}$] Effective channel gain of user $k$ on subcarrier $n$;
\item[$\check{d}_k$] Minimum rate required by the real-time user $k$.
\end{description}
The decision variable are defined as
\begin{description}
\item[$p_{n,k}$] Power assigned to user $k$ on subchannel $n$.
\end{description}
We also define
\begin{description}
\item[$r_{n,k}$]  The transmission rate achieved by user $k$ on subchannel $n$ for one OFDM symbol.
\end{description}
We assume that capacity achieving channel coding is employed on each subcarrier such that
\begin{equation}
  r_{n,k} = \log_2 \left( 1 + p_{n,k} \right). \label{eq:defrate}
\end{equation}

The objective of the PA algorithm is to find the $p_{n,k}$ to maximize the weighted sum
rate of the users subject to the power and minimum rate
constraints. The users scheduling weights {$c_k$} and the minimum rate
constraints are determined by a higher layer scheduler. The
optimization problem is formulated as follows
\begin{align}
  \label{eq:sub-problem-q}
  \max_{ p_{n,k} }  \quad \sum_{n=1}^N \sum_{k=1}^K  c_k
  & \log_2( 1 + p_{n,k} )  \\
  \sum_{n=1}^N \sum_{k=1}^K  \beta_{n,k} \,\ p_{n,k}
  & \le   \check{P} \label{eq:subp-q-pow} \\
\sum_{n=1}^N \log_2( 1 + p_{n,k} ) & \ge  \check{d}_k
\quad k \in { \mathcal{D} } \label{eq:subp-q-rat} \\
p_{n,k} & = 0, \quad \forall n, k\notin S_n \\
p_{n,k} & \geq 0, \quad \forall n, k\in S_n.  \label{eq:subp-box}
\end{align}
Constraint (\ref{eq:subp-q-pow}) is the total power constraint and
constraints (\ref{eq:subp-q-rat}) are the RT users' minimum rate
constraints. Note that this
problem is convex
since it maximizes a concave function over  the convex set defined
by~(\ref{eq:subp-q-pow}--\ref{eq:subp-box}).
The solution  also has the following property.
\begin{lem}
\label{lem:heur-pow-contr-eq} The solution  of
problem~(\ref{eq:sub-problem-q}--\ref{eq:subp-box})  satisfies the
power constraint (\ref{eq:subp-q-pow}) with equality.
\end{lem}
\begin{proof}
This is a straightforward consequence  of the fact that the
rate function~(\ref{eq:defrate}) is an increasing function of
$p_{n,k}$.
Suppose $\vect{p}^{(a)}$ is a feasible point
that
satisfies (\ref{eq:subp-q-pow}) with \emph{strict} inequality and satisfied (\ref{eq:subp-q-rat}).
Then there exists a
$\boldsymbol\Delta \geq 0$  such that both the power
constraint~(\ref{eq:subp-q-pow}) and the rate constraints~(\ref{eq:subp-q-rat})
are feasible for
$\vect{p}^{(b)}= \vect{p}^{(a)} + \boldsymbol\Delta  $. But in this
case,
the values of $r_{n,k}$  will also increase so that
$\vect{p}^{(a)}$ cannot be optimal.
\end{proof}
\noindent In all that follows, we will thus assume that~(\ref{eq:subp-q-pow}) is an
equality.

We can solve~(\ref{eq:sub-problem-q}--\ref{eq:subp-box}) in a
number of ways, for instance,  with a general-purpose nonlinear
programming solver to get  solution of the primal, or use a
Lagrangian decomposition method and solve the dual by a
subgradient technique. However, this is much too long for use in
real time applications and  we thus need fast approximations. The
ones we propose in this paper are based on a fixed-point
reformulation of the problem in the dual space.

Because the problem is convex, we can solve it optimally by computing
a solution to  the first-order optimality equations.
Let first
define the  Lagrange multipliers $\theta \ge 0 \in \mathbb{R}_+^K$ for the power constraint
(\ref{eq:subp-q-pow}) and  $\delta_k \ge 0 \in \mathbb{R}_+^K$ for
the rate constraints (\ref{eq:subp-q-rat}).  We then write the partial
Lagrange function
function
\begin{align}
\label{eq:sub-problem-lagr}
\iftoggle{onecolumnformat} {
 \mathcal{L}_2( \vect{p}, \theta,\boldsymbol{\delta}) =  - \theta \check{P}
   + \sum_{k=1}^K \delta_k \check{d}_k   + \sum_{n=1}^N \sum_{k=1}^K -( c_k + \delta_k )
   \log_2( 1 + p_{n,k} ) + \theta \beta_{n,k} p_{n,k} - \sum_{n,k}
   \mu_{n,k} p_{n,k}
} %
{
 \mathcal{L}_2( \vect{p}, \theta,
\boldsymbol{\delta}) &=  - \theta \check{P}
   + \sum_{k=1}^K \delta_k \check{d}_k   + \sum_{n=1}^N \sum_{k=1}^K -( c_k + \delta_k )
\\ \nonumber
& \log_2( 1 + p_{n,k} ) + \theta \beta_{n,k} p_{n,k} - \sum_{n,k}
   \mu_{n,k} p_{n,k}
}
\end{align}
where $\boldsymbol{\delta}$ is the vector of variables
$\delta_k$. For convenience we have defined
$\delta_k$ for $k \notin { \mathcal{D} }$ for which we set
$\check{d}_k=0$.
We get the first-order optimality conditions
\begin{equation}
  \frac{\partial \mathcal{L} }{ \partial p_{n,k} } = - \frac{ c_k +
    \delta_k}{1 + p_{n,k} } + \theta \beta_{n,k}  = 0  \label{eq:firstorder}
\end{equation}
which yield the optimal water-filling power allocation as a function of the
multipliers $\theta$ and $\boldsymbol{\delta}$
\begin{align}
\label{eq:optp-def1}
\overline{p}_{n,k} =& { \left[ \frac{ c_k+
\delta_k } { \theta \beta_{n,k} \ln2 } -1 \right] }^+, \quad
\forall n, \forall k: \beta_{n,k} \neq 0.
\end{align}
From this, we get  the optimal rate allocation
\begin{align}
\label{eq:rate-computation} r_{n,k} =& \log_2 \left( 1 + { \left[
\frac{ c_k+ \delta_k } { \theta \beta_{n,k} \ln2 } -1 \right] }^+
\right), \quad \forall n, \forall k: \beta_{n,k} \neq 0.
\end{align}
We also define the set of strictly positive allocations for user $k$
\begin{equation}
\mathcal{B}_k ( \theta, \delta_k) =  \left\{ n \vert
\overline{p}_{n,k} > 0 \right\}. \label{eq:defbctheta}
\end{equation}
From~(\ref{eq:optp-def1}), we see that the size $\sigma_k(\theta, \delta_k)$ of
this set increases with $\delta_k$ and decreases with $\theta$.
We have from~(\ref{eq:optp-def1}) an expression for the solution in
terms of the multipliers. The standard procedure is to
replace this  into the
constraints and solve for the multipliers. This generally yields
a set of nonlinear equations and in the case of inequality
constraints, it has the added complexity of choosing the set of
constraints that are tight at the solution.

We now reformulate the problem in terms of a system of fixed-point
equations. This will be used later to get fast approximations and it
also yields an algorithm for computing a feasible solution. We can do
this by replacing~(\ref{eq:optp-def1})
in~(\ref{eq:subp-q-pow}) and~(\ref{eq:subp-q-rat}) and get
\begin{align}
  \delta_k & \ge  \left[ \theta \ln2
  \left( 2^{\check{d}_k}
    \prod_{n \in \mathcal{B} (\theta, \delta_k) }
    \beta_{n,k} \right)^{ \sigma(\theta, \delta_k)^{-1} }  - c_k \right]^+
  \label{eq:delta1theta}\\
\theta  & = \frac{
\sum_{k=1}^K \sigma_k( \theta, \delta_k )  ( c_k + \delta_k ) }
{ \left( \check{P} + \sum_{k=1}^K \sum_{n\in\mathcal{B}_k(\theta, \delta_k ) }
  \beta_{n,k} \right)\ln 2 }
\label{eq:theta-maxthr-heur-comp} .
\end{align}
This is a set of fixed-point equations because of the presence of
$\mathcal{B}(\theta, \delta_k)$ in the right-hand side.
Any set of variables $\theta$, $\boldsymbol{\delta}$ that solve these
equations will yield an optimal solution via~(\ref{eq:optp-def1}).
Also note that if $\mathcal{B}$ is fixed, each $\delta$ is a linear
function of $\theta$, and $\theta$ is a linear function of
$\boldsymbol{\delta}$. This means that both $\theta$ and
$\boldsymbol{\delta}$ can be viewed as piecewise linear functions of
the other multipliers.

For now, we want
to solve problem~(\ref{eq:sub-problem-q}--\ref{eq:subp-box})
efficiently. The first approach is to take advantage of the fact that
we expect that the number of real-time users will generally be
small. This means that if we solve the power
allocation \emph{without} the rate constraints~(\ref{eq:subp-q-rat}), unless the real-time users happen to have particularly
bad channels,
there is a good chance that the constraints will be met
automatically, so that this solution is optimal. We study this problem in section~\ref{subsec:max_thr_heur}.

If this solution does not meet the rate constraints,
we incorporate them into the problem and solve
it approximately. We describe in
section~\ref{subsec:rat_con_heur} two heuristic methods for this latter problem. The first
one is iterative and yields a feasible solution where the rate
constraints are tight. The second one is
based on
adjusting the variables to make the solution feasible.
It requires no iterations and often finds a
feasible point that satisfies the power constraint with equality and the rate
constraints with inequality.
\section{Power   Allocation without Rate Constraints}
\label{subsec:max_thr_heur}
We  want in this section to compute quickly the optimal solution of
problem~(\ref{eq:sub-problem-q}--\ref{eq:subp-q-pow}) and~(\ref{eq:subp-box})  without the rate
constraints~(\ref{eq:subp-q-rat}). In this case,  we have $\delta_k =
0$ but because it will  be used in
section~\ref{subsec:rat_con_heur} when computing a feasible solution, we present the algorithms with an arbitrary, but fixed, value
of $\boldsymbol{\delta}$.
We compare two methods that yield an exact solution. The first one
finds a root of the power constraint and is guaranteed to converge in a fixed number
of iterations. The second method does not
have a guaranteed convergence but is faster that the root-finding method.
\subsection{Root of the Optimality Conditions}
\label{sec:solv-optim-cond}
This method computes the optimal value of $\theta$ by finding a zero of
the first-order optimality condition.
We replace~(\ref{eq:optp-def1})
in~(\ref{eq:subp-q-pow}) and solve for $\theta$ the nonlinear equation
\begin{equation}
\label{eq:heur-pow-const2} g_1(\theta) = \sum_{k=1}^K \sum_{n=1}^N \beta_{n,k}
\left[ \frac{ ( c_k + \delta_k ) } { \theta \beta_{n,k} \ln2 } -1
\right]^+ = \check{P}.
\end{equation}
The function $g_1(\theta)$ is
continuous but not differentiable at the \emph{corner} points
$\theta_{n,k}$   defined as
\begin{equation}
  \theta_{n,k} = \frac{ ( c_k + \delta_k ) }{\beta_{n,k} \ln2}.
\end{equation}
For a given value of $\theta$, $g_1$  is the sum of all hyperbolic
segments $\overline{p}_{n, k} ( \theta)  $ such that $\theta \le \theta_{n,k}  $. A typical
function is shown on Fig.~\ref{fig:solnonlin}.
%
\begin{figure}
  \centering
  \includegraphics[scale=0.7]{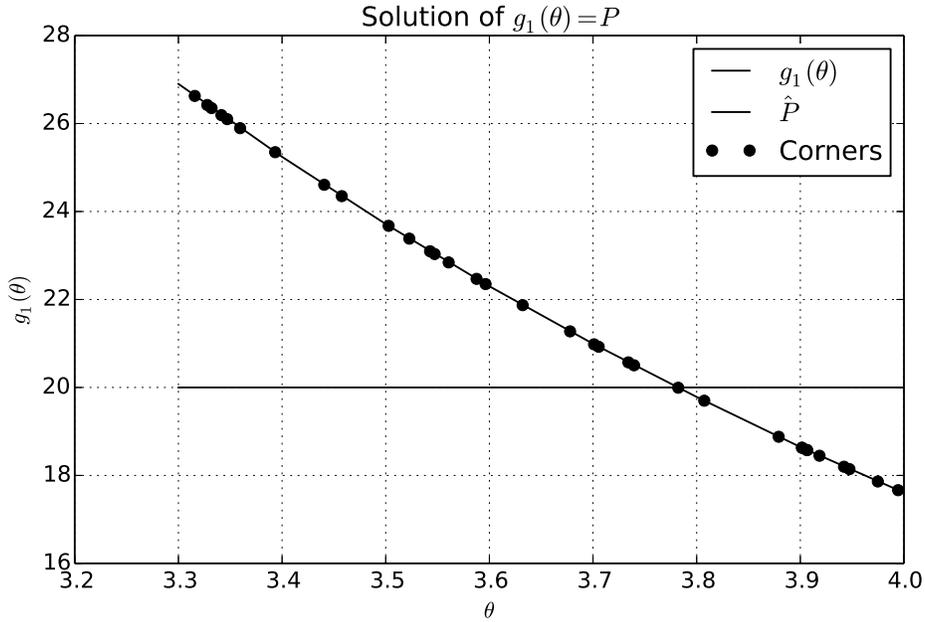}
  \caption{Solving the first-order equations}
  \label{fig:solnonlin}
\end{figure}
We could solve it
using a general-purpose zero-finding algorithm for
non-differentiable functions but we can take advantage of the
particular structure of the problem to design a more efficient
algorithm. This is based on the method presented in~\cite{palomar-03} that takes
advantage of the finite number of corner points.
First, we relabel the corner points in increasing values
$\theta_i, i=1, \ldots NK$. The sets $ \mathcal{B}_k(\theta,
\delta_k)$   can be defined in a more compact way
\begin{equation}
  \mathcal{B}_k(\theta, \delta_k ) = \left\{  n \vert \theta \le
    \theta_{n,k}  \right\} .  \label{eq:defb}
\end{equation}
Note that $\mathcal{B}_k(\theta, \delta_k )$ is constant inside any given
interval $[\theta_{i-1}, \theta_i]$.

First we  quickly find the interval $\Delta$ where the solution lies by doing
a binary search on the set of $\theta_i$. At iteration $m$, denote the
two indices that define the
interval that contains the solution by $\left( i_1^{(m)}, i_3^{(m)}
\right) $ so that
$\Delta^{(m)} = [\theta_{i_1^{(m)} } ,
\theta_{i_3^{(m)} } ]$. Obviously, $i_1^{ (0) } = 1$ and
$i_3^{  (0) } = NK $. Let $i_2^{(m)}$ be the index
of the  middle point $\lfloor \left( i_3^{(m)} -  i_1^{(m)}
\right) / 2 \rfloor$.
If $g_1(\theta_{i_2^{(m)} }) > \hat{P}$, then
$i_1^{(m+1)} = i_2^{(m)}$ and $i_3^{(m+1 )} =
i_3^{(m)}$. If
$g_1(\theta_{i_2^{(m)} } ) < \hat{P} $, then
$i_3^{(m+1)} = i_2^{(m)}$ and $i_1^{(m+1 )} =
i_1^{(m)}$.
The number of function
evaluations is bounded by $\log_2(NK)$ since the size of the interval
is divided by 2 at each iteration and the algorithm converges in a finite
number of iterations. We also need to compute the
$\theta_i$ and sort them.

Once we know that the solution lies in
the interval $[\theta_{i-1}, \theta_i]$, we can compute  the value of the
sets $\mathcal{B}_k(\theta, \delta_k)$ and find the value of $\theta $
directly from~(\ref{eq:theta-maxthr-heur-comp}).



\subsection{Fixed-Point Algorithm}
\label{sec:fixed-point-algor}
We can also find the solution of~(\ref{eq:theta-maxthr-heur-comp}),
which we write $\theta = g_2(\theta)$,
by repeated substitution of the left-hand
side $\theta$  into $g_2$.
We first compute an initial value of $\theta$ by replacing the $c_k$'s and
$\beta_{n,k}$'s by their average values in~(\ref{eq:optp-def1}) and
setting the condition
$\overline{p} = 0$.  Using the current value of $\theta$, we then compute the
  $\mathcal{B}_k(\theta, \delta_k )$s. A new value of $\theta$ can then be found from the right-hand side
  of~(\ref{eq:theta-maxthr-heur-comp}) and we iterate until the new value of $\theta$ is the same
  as the previous one.
%
\begin{figure}
  \centering
  \includegraphics[scale=0.7]{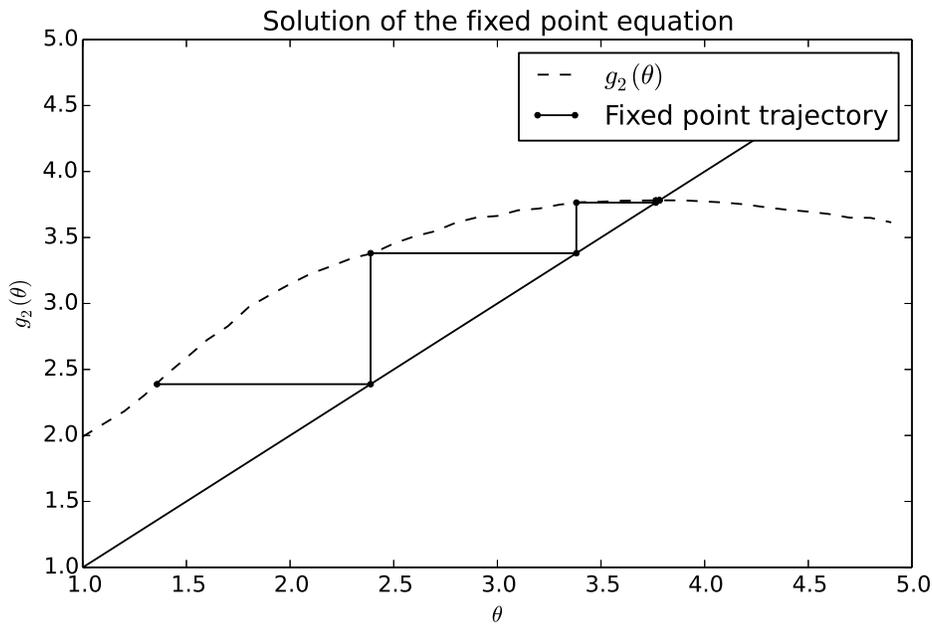}
  \caption{Fixed point iterations}
  \label{fig:fixdpt}
\end{figure}
The stopping rule is an equality since various values of
$\theta$ produce only a finite number of sets $\mathcal{B}_k$ so that
the only values produced by the iterations are the $\theta_i$. We can
see the iteration procedure on Fig.~\ref{fig:fixdpt} where we have
shown the left-hand side in the form $y = \theta$ and the right-hand
side $y = g_2(\theta)$. The solution is at the intersection of the two
curves. We also show the points computed by the substitution
algorithm.

The main advantage of this method is that there is no need to compute
the $\theta_i$ and to sort them in advance.
The problem with substitution methods is that it is
difficult to guarantee convergence. But, in practice we have found that
the fixed-point method always converged for all realistic values  of
the parameters.

We  compare the cpu time of the four solution techniques in
Table~\ref{tab:cpufp}. We have fixed $ N=16$, $\hat{P} = 5$, $c_k =
1$ and generated the $\beta_{n,k}$ from a simple but realistic channel
model.  These values were obtained on a 2.5 GHz Intel
Core i5 computer. We have computed the solution 100 times in each case to have a
reasonably stable estimate. We can see that the fixed-point method is
the fastest of all methods and the number of iterations does not
change much with the problem size.
The option is then between a technique with guaranteed convergence but
somewhat slower than the other technique with no strict convergence guarantee. A compromise could
be to use the fixed-point method and switch to the binary search
technique if it has not converged within a small number of iterations, say 10.

\begin{table}
  \centering
  \begin{tabular}{c c c c c c c c}
    M &  Minos &  \multicolumn{2}{c}{Nonlinear} & \multicolumn{2}{c}{Binary} &  \multicolumn{2}{c}{Fixed Point} \\
       &  Cpu msec & Cpu msec & No iterations & Cpu msec   & No iterations & Cpu msec   & No iterations \\
\hline
 2& ~23.3 & 19.0 & 20  & ~1.5 & 0 & 2.0 & 3  \\
 4& ~71.8 & 43.8 & 21  & 12.6 & 6 & 4.2 & 3  \\
 8& 292.~ & 92.5 & 22  & 31.4 & 9 & 7.6 & 3  \\
  \end{tabular}
  \caption{Computation load without rate constraints}
  \label{tab:cpufp}
\end{table}

\section{Power Allocation  with Rate Constraints }
\label{subsec:rat_con_heur}
The problem of allocating power with rate constraints is solved in two
steps. First, we try to allocate power without the rate constraints
using one of the algorithms of section~\ref{subsec:max_thr_heur} since
this is very fast.
For the case where some of the RT users do not get their minimum rate,
we propose in this section two algorithms that try to find a solution that is feasible
both for the rate and power constraints, even though it may not be optimal.  First we present an
iterative algorithm that is guaranteed to produce a feasible solution with the rate
constraints at their lower bounds. We then propose a non-iterative
procedure that is faster  but which is not  guaranteed to find a feasible solution.
\subsection{Feasible Solution at Bounds}
\label{sec:feasible-solution-at}
\begin{figure}
  \centering
  \includegraphics[scale=0.7]{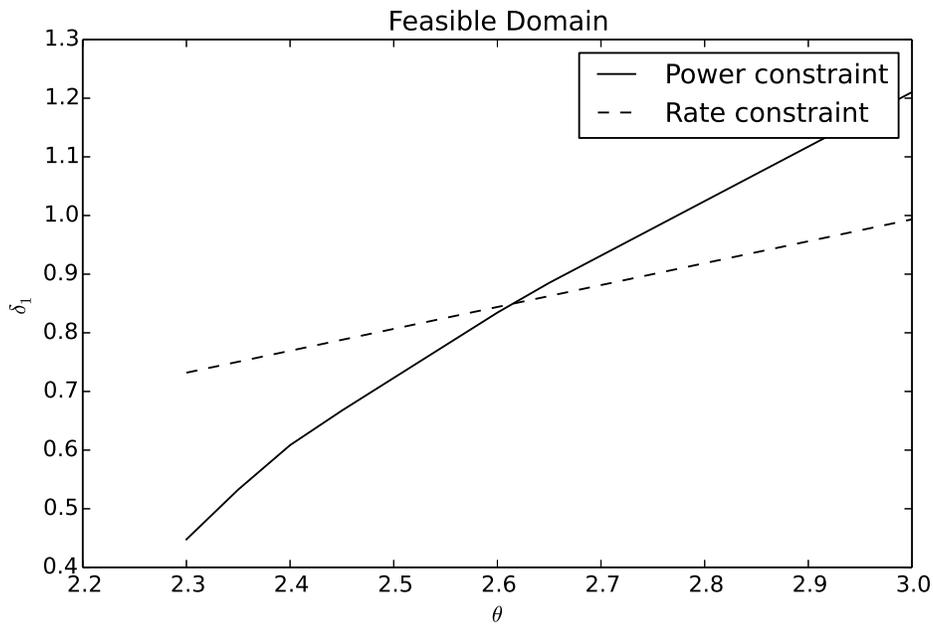}
  \caption{Domain for a single rate constraint}
  \label{fig:domain}
\end{figure}
To simplify the discussion,
consider first a problem with a single rate constraint, say for
user~1. The boundaries of the two
constraints~(\ref{eq:subp-q-pow}--\ref{eq:subp-q-rat}) each
define an implicit function $\delta_1(\theta)$. We can see an example
of these boundaries on Fig.~\ref{fig:domain} for a small network
with $K = 2$, $N=8$, $P=2$ and a single real-time user. The set of feasible
solution  is the
region above the rate constraint and below the power constraint.

A remarkable feature of this plot is that the boundary of the rate
constraint seems to be linear.  We already know that the boundary is
piecewise linear but we now
prove the following lemma.
\begin{lem}
  The boundaries of the rate constraints are linear.
\end{lem}
\begin{proof}
If $\theta \approx 0$, all the $p_{n,k}  \approx \infty $ and~(\ref{eq:subp-q-rat}) will be a
strict inequality. From the complementarity condition, we will then
have $\delta_k = 0$. This is consistent with the interpretation of
$\theta$ as a penalty for the power constraint: If $\theta$ is small,
there is no penalty for exceeding the power constraint so that there is
a lot of power available.  We can  then choose large values of the
$\overline{p}_{n,k}$
to increase the total  rate so that the rate constraint will be an
inequality.

As $\theta$ increases, there will be a point $\overline{\theta}_k$
where each rate $k$
constraint will become tight. This is given by
\begin{equation}
  \overline{\theta}_k = \frac{c_k \sum_{n \in \mathcal{B}_k(\overline{\theta}_k, 0) }
    \beta_{n,k} } { \check{d}_k \log
    2 } . \label{eq:theta0}
\end{equation}
At that point, we must increase
$\delta_k \ge 0$ to stay on the boundary.
For all values of
$\theta \ge \overline{\theta}_k$, we choose a value of $\delta_k$ which is a
linear function  of $\theta$ with slope
$m_k$
\begin{align}
  \delta_k & = \left[ m_k \theta - c_k \right]^+  \label{eq:deltalin}\\
  m_k & = \ln 2   \left( 2^{\check{d}_k}
    \prod_{n \in \mathcal{B} (\theta, \delta_k) }
    \beta_{n,k} \right)^{ \sigma(\theta, \delta_k)^{-1}
  } \label{eq:slopes} .
\end{align}
With this choice of $\delta_k$, we know that the rate constraint in an
equality over the whole range $\theta \ge \theta_k$ and from this, we see that
\begin{align*}
  \overline{p}_{n,k} & = \left[ \frac{c_k + \delta_k}{\theta \log 2 \beta_{n,k}}
    - 1 \right]^+ \\
  & = \left[ \frac{c_k + m_k \theta - c_k}{\theta \log 2 \beta_{n,k}}
    - 1 \right]^+ \\
  & = \left[ \frac{ m_k }{ \log 2 \beta_{n,k}}     - 1 \right]^+
  \end{align*}
which is independent of $\theta$. In other words, once $\theta \ge
\overline{\theta}_k$, the set $\mathcal{B}_k(\theta, \delta_k)$ does not change
and the rate constraint boundary is linear.
\end{proof}

We see from Fig.~\ref{fig:domain} that the point $(\theta^{(1)}, 0)$ where the power constraint
boundary would meet the $\theta$ axis, is outside the feasible domain.
The problem is to quickly find a point  both above the rate boundary
and below the power boundary and, if possible, a good point.

We make use of the linearity of the boundary to compute a feasible solution
quickly by computing the intersection of the rate and
power boundaries. First, we compute the slopes $m_k$ of the rate boundaries
from~(\ref{eq:slopes}) at the point $(\theta^{(1)}, 0)$. Given any value
of $\theta$, we can then compute all the $\delta_k$'s from~(\ref{eq:deltalin}) and from this,
the value of the total power $P(\theta)$  and total rate $R(\theta)$
from~(\ref{eq:optp-def1}). We can see an example of these curves on
Fig.~\ref{fig:power-rate-bound}.
\begin{figure}
  \centering
  \includegraphics[scale=0.7]{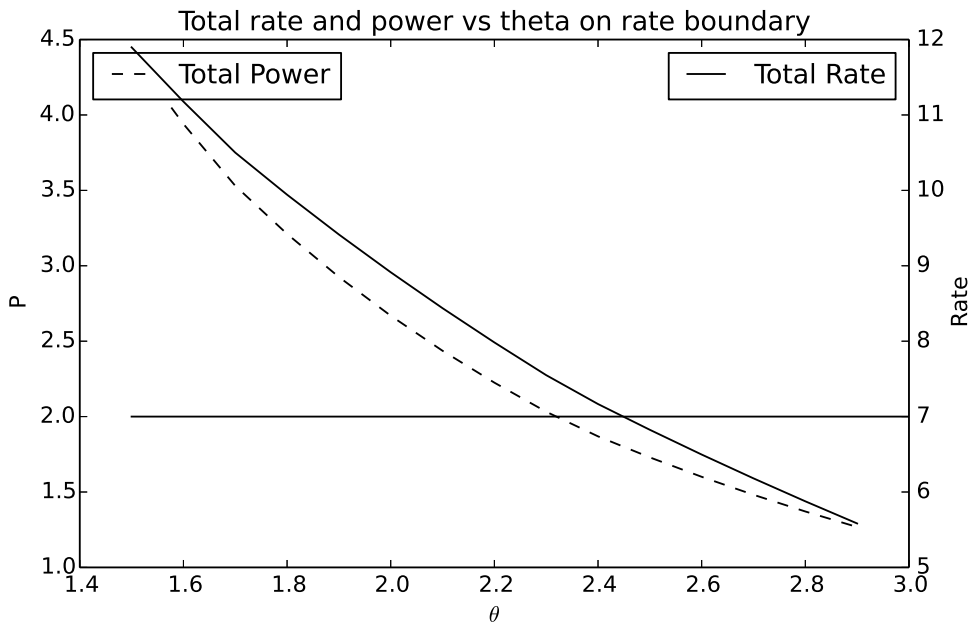}
  \caption{Rate and Power Curves vs $\theta$ on the Rate Boundaries}
  \label{fig:power-rate-bound}
\end{figure}
We compute the solution by solving $P(\theta) = \hat{P}$ by some
numerical root-finding algorithm which is guaranteed to converge
asymptotically. This is the best we can do if we
impose the conditions that the rate constraints are met with equality
since moving to the right on these boundaries means that the power
constraint is not tight, which is a condition for an optimal
solution.

\subsection{Fast Heuristic}
\label{sec:previous-heuristic}
As we will see later, the time needed to compute a solution on the
boundary may turn out to be too large for use in real time. Instead,
we propose  the following simple, non-iterative
algorithm. The algorithm is based on the observation that a standard way
to satisfy some constraint $k$  is to increase the
corresponding multiplier $\delta_k$. From~(\ref{eq:optp-def1}), we see that this will
increase the $\overline{p}_{,nk}$ which will make the constraints more
feasible. But this will also make the power constraint infeasible so
that we need to increase $\theta$ as well. Let us define
\begin{description}
\item[$\theta^{(1)}$] the multiplier from the solution
  without rate constraints
  \item [$\overline{\theta}$] an upper bound on $\theta$ given by
$\overline{\theta} = W \theta^{(1)}$ with $ W >1$.
\end{description}
Starting from $\theta = \theta^{(1)}$, the heuristic basically moves a
fixed amount in the direction of increasing $\theta$. At that point,
it computes the solution on the rate boundaries and then tries to
improve $\theta$ one more time as follows:
\begin{enumerate}
\item Compute the set of unsatisfied users $\mathcal{T}  = \left\{ k
    \suchthat r_k < \check{d}_k \right\} $.
   \item Compute $\mathcal{B}_k(\overline{\theta}, 0)$
\item For users $ k \not \in \mathcal{T}$, set
$\delta_k^{(2)}=0$. For the other users, compute $\delta_k^{(2)}$
from~(\ref{eq:delta1theta}) with $\theta = \overline{\theta} $ and
$\mathcal{B}_k = \mathcal{B}_k(  \overline{\theta},
\delta_k^{(2)}  )$.
%
%
\item Compute a new
  $\theta^{(2)}$ from~(\ref{eq:theta-maxthr-heur-comp}) with
  $\boldsymbol{\delta} = \boldsymbol{\delta}^{(2)}  $ and $\mathcal{B}_k =
  \mathcal{B}(\overline{\theta}, \delta_k^{ (2) } ) $.
\end{enumerate}
\subsection{Improved Heuristic} \label{sec:tighter-bounds}
The single iteration heuristic uses an arbitrary value $W$ for the
step size which may not produce a very good solution.
We can improve the results if we  choose the parameter $W$
based on the difference between the required rate and the
rate achieved by the maximum throughput PA. This increases the
chance that the step will lead to a point to the right of the
intersection of the two constraint curves.
First we compute a
different value of $W$ for each real-time user
\begin{equation}
W_k=2^{\epsilon(\check{d}_k-r_k)}. \label{eq:eps-def}
\end{equation}
and then use the maximum $W_k$ among all RT users
\begin{align}
\label{eq:bar-epsilon} W &= \max_{ k \in \mathcal{T} } W_k  \\
 \overline{\theta} &= W \theta^{(1)} .
\end{align}
This way, we can increase the emphasis we want to put on
meeting the rate constraints by increasing the value of $\epsilon$.
\subsection{Infeasibility}
\label{sec:infeasibility}
\begin{figure}
  \centering
  \includegraphics[scale=0.7]{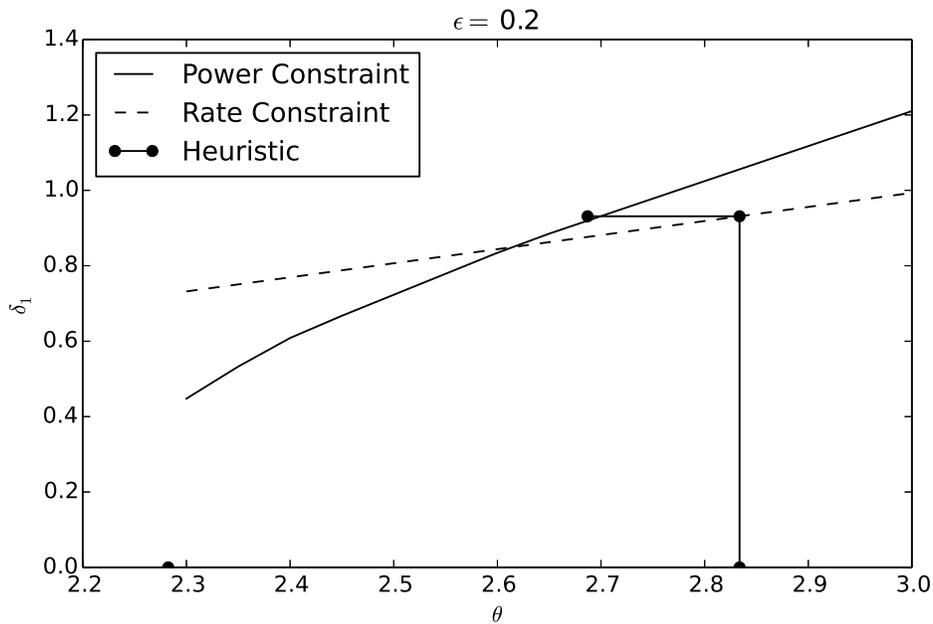}
  \caption{Fast Heuristic, Feasible Solution}
  \label{fig:heuristicfeasible}
\end{figure}
We can see on Fig.~\ref{fig:heuristicfeasible} the path of the heuristic
in the $(\theta, \delta)$ space with $ \epsilon = 0.2 $. In this case, the final point is
feasible since it lies above the rate boundary and lies on the power
boundary. We say that the algorithm has succeeded. On the other hand,
it is possible for the algorithm not to succeed for a variety of
reasons.
A practical approach could then be to try the fast heuristic first. If
it fails, one could either try the zero-finding solution with a
limited number of iterations or simply use the infeasible solution.

\begin{figure}
  \centering
  \includegraphics[scale=0.7]{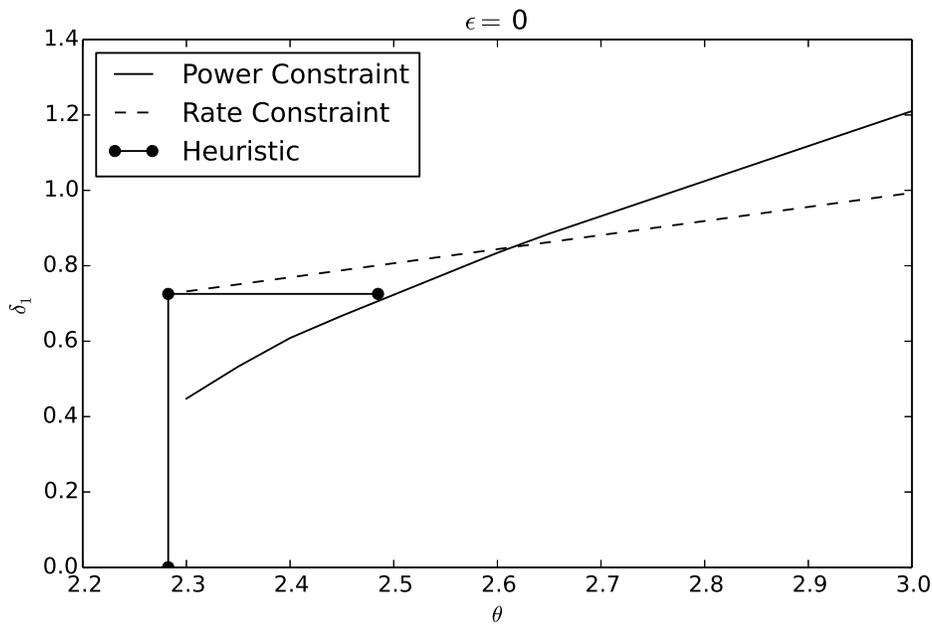}
  \caption{Fast Heuristic, Infeasible Solution}
  \label{fig:heuristicinfeasible}
\end{figure}

The first one is that we don't know where the intersection of the two boundaries lie
so that we don't know how far to the right we must go to. Choosing a value
of $ \epsilon $ too small will not go far enough and the  final point will not
be feasible. This can be seen on Fig.~\ref{fig:heuristicinfeasible}
where we have chosen $ \epsilon = 0 $. The point $(\overline{\theta}, 0)$ is
to the right left  of the intersection point of the two boundary
curves so that moving to the right until we reach the power boundary
still leaves us below the rate constraint.  In this case, the
algorithm has failed.

There is another, more subtle reason, why the algorithm may
fail. Consider the final point of
figure~\ref{fig:heuristicfeasible}. We can see that the point is
slightly to the left of the power boundary.
The reason is that we compute
$\theta^{(2)}$ using $\mathcal{B}_k(\overline{\theta},
\delta_k^{ (2) } )$ which is not the same as $\mathcal{B}_k(
\theta^{(2)}, \delta_k^{ (2) } ) $, hence the small error.

Another potential problem is that by increasing $\theta$ from
$\theta^{(1)}$,
the heuristic algorithm reduces the power of the users that are
not in $\mathcal{T}$. If some of these are RT users and have a rate
close to $\check{d}_k$, they may become unfeasible.
A simple solution is to add these users to $\mathcal{T}$ by increasing their
$\check{d}_k$ and run the fast heuristic again.

%

We can see that the algorithm
produces a feasible point as follows.  We assume that users close to the
unfeasibility boundary are all included in set $\mathcal{T}$ by
increasing their $\check{d}_k$.
\begin{lem}
\label{theo:heur-feas} When successful, the fast algorithm
produces a feasible point to
problem~(\ref{eq:sub-problem-q}--\ref{eq:subp-q-rat}).
\end{lem}
\begin{proof}
By construction, the final point lies on the power boundary. It also
lies above the rate boundary  for the
following reason. The point $(\overline{\theta},
\delta_k^{ (2) } ) $ lies on the rate boundary by construction. Given that the slope of
the rate boundary is positive, all the poinst $(\theta,
\delta_k^{ (2) } ) $ with $\theta < \overline{\theta}$ will also
lie above the boundary.
\end{proof}
Note also that
this method can be used with an arbitrary method for computing the
unconstrained power allocation, such as~\cite{ling12,yu-boyd-04,jang-03}, as
long as we get the value of $\theta^{(1)}$
from the solution.
\section{Numerical Evaluation}
\label{sec:results}
In this section, we compare the three different
methods to solve
problem~(\ref{eq:sub-problem-q}--\ref{eq:subp-box}): An exact solution
by a nonlinear solver~\cite{minos93}, the boundary solution of~\ref{sec:feasible-solution-at} and the fast
heuristic of~\ref{sec:tighter-bounds}. We look at their accuracy and
the cpu time needed to compute a solution.  The nonlinear equations
were solved by the netlib function \verb|bisect| through the Scipy
Python interface. All results were obtained on an Intel Core i5 cpu
running at 2.5 GHz.
\subsection{Generating Test Cases}
\label{sec:gener-test-cases}
Each test case is defined first by the set $(K,N,M,P)$.
We use a Rayleigh fading model to
generate the user channels such that each component of the channel
vectors $h_{k,n}$ are i.i.d. random variables distributed as
$\mathcal{CN}(0, 1)$. We also assume independent fading between
users, antennas and subcarriers.  We assign the subchannels with
a variant of the semiorthogonal
user selection (SUS)  heuristic~\cite{yoo06} and then compute the $\beta_{n,k}$ from
this.
After this, we select the number of real-time users $R$. The
set of values we have used is shown in Table~\ref{tab:paramtests}.
\begin{table}
  \centering
  \begin{tabular}{cccccc}
    Number & $K$ & $N$ & $M$ & $\hat{P} $ & $R$ \\
    \hline
    \refstepcounter{Problemset}\label{pr2025253} \ref{pr2025253}& 20 & 25 & 2 & 5 & 3 \\
    \refstepcounter{Problemset}\label{pr20252510} \ref{pr20252510}& 20 & 25 & 2 & 5 & 10 \\
    \refstepcounter{Problemset}\label{pr8025850} \ref{pr8025850}& 80 & 25 & 2 & 5 & 0 \\
    \refstepcounter{Problemset}\label{pr80258510} \ref{pr80258510}& 80 & 25 & 2 & 5 & 10 \\
    \refstepcounter{Problemset}\label{pr80258520} \ref{pr80258520}& 80 & 25 & 2 & 5 & 20 \\
    \refstepcounter{Problemset}\label{pr80258530} \ref{pr80258530}& 80 & 25 & 2 & 5 & 30 \\
    \refstepcounter{Problemset}\label{pr80258540} \ref{pr80258540}& 80 & 25 & 2 & 5 & 40 \\
    \refstepcounter{Problemset}\label{pr80258550} \ref{pr80258550}& 80 & 25 & 2 & 5 & 50 \\
    \refstepcounter{Problemset}\label{pr80258560} \ref{pr80258560}& 80 & 25 & 2 & 5 & 60 \\
    \hline
  \end{tabular}
  \caption{Parameters for Test Cases}
  \label{tab:paramtests}
\end{table}

The next step is to generate the rate bounds in such a way
that the problems are known to be feasible. Also, we
would like to control the ``difficulty'' of the problem to see how
much this impacts on the accuracy and cpu time. First, we maximize the
total rate of real-time users by solving
problem~(\ref{eq:sub-problem-q}--\ref{eq:subp-q-pow}) over the
$p_{n,k},\, k \in \mathcal{D} $ only and  without the rate
constraints. Call this problem $P_0$ and $R_0$ the total rate of the
real-time users.
We then know there is enough power to give the
real-time users the rates $r_k^{ (0) } $ produced by the unconstrained
solution.

This means that any  problem $P_1$ with rate
constraints $\check{d}_k = r_k^{ (0) } $ has a
feasible solution  with only
the real-time users having a positive rate and with the rate
constraints at their bounds.    If this were not the case, there would
be some spare power that could be used to construct a solution to $P_0$ with a total rate $R_1 >
R_0$ which contradicts the fact that $R_0$ is the optimal value for $P_0$.

We also know that any
problem with rate constraints $\check{d}_k = s  r_k^{ (0) } $ is
feasible if $ 0 \le s \le 1$.  Furthermore, the value of $s$ is a
rough measure of the ``tightness'' of the constrained data set. The
problem is unconstrained when $s = 0$ and is tightly constrained when
$s = 1$ where only the real-time users can get some power.  The scale $s$
will be the primary parameter used to display results.
\subsection{Effect of $\epsilon$}
\label{sec:effect-epsilon}
\begin{figure}
  \centering
  \includegraphics[scale=0.7]{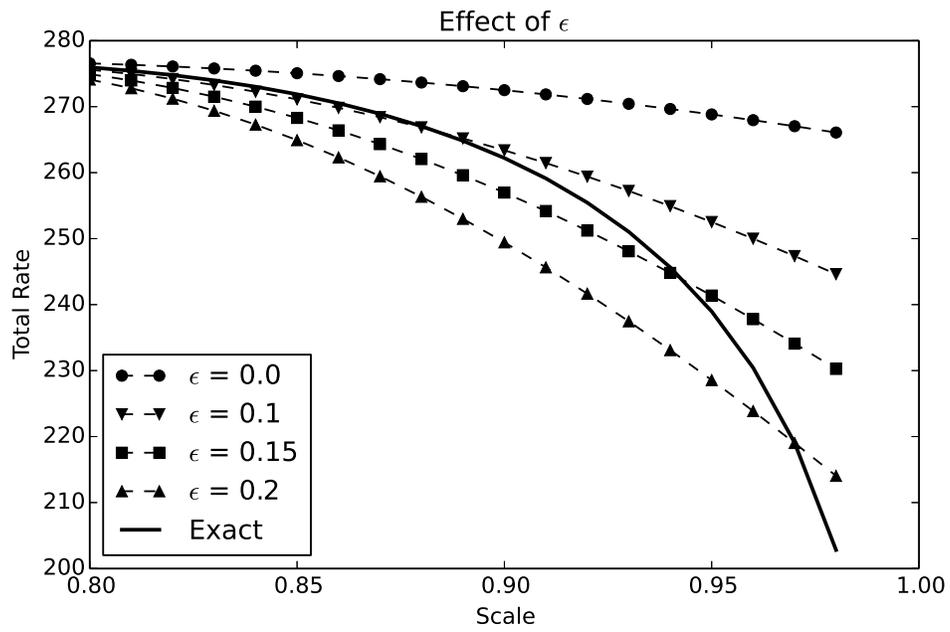}
  \caption{Effect of $\epsilon$ on Heuristic}
  \label{fig:heurvseps}
\end{figure}
First we examine the effect of $\epsilon$ on the accuracy of the fast
heuristic for problem~\ref{pr2025253}. This is show on Fig.~\ref{fig:heurvseps} where we plot the objective
value computed by the heuristic for different values of $\epsilon$ and
also the optimal value computed by the nonlinear solver as a function
of the scale parameter. For each curve corresponding to the heuristic,
we can define a range of values of $s$ where the heuristic value is
below the optimal rate and one where it is above. In this latter case,
this is an indication that the solution computed by the heuristic is
not feasible. The curves show that there is a clear tradeoff between
the accuracy of the solution and the range of problems where the
solution is feasible. For small values of $\epsilon$, the range of
problems where the heuristic can produce a feasible solution is quite
small but the solutions are quite accurate. This is opposite to the
curves with large $\epsilon$ where the algorithm can compute feasible
solutions over a larger range of problems but where the objective function solution is
not as good.

\subsection{Solution Accuracy}
\label{sec:accuracy-solutions}
\begin{figure}
  \centering
  \includegraphics[scale=0.7]{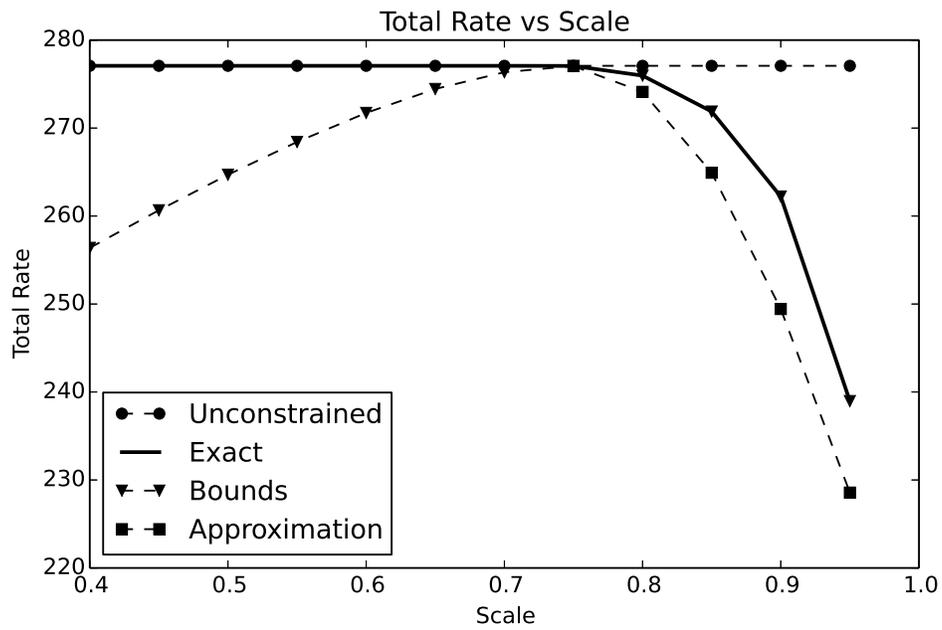}
  \caption{Accuracy of Heuristics, Small Network, Few Real-Time Users}
  \label{fig:objvsscale-20-P1}
\end{figure}
Next we show the accuracy of the two heuristics by comparing the total
rate they produce with the exact value as a function of the scale
parameter. In all cases, we choose $\epsilon = 0.2$. The first case is
for problem~\ref{pr2025253} shown on Fig.~\ref{fig:objvsscale-20-P1}. Here, it
turns out that all the rate constraints are at their bound in
the optimal solution. This is not unexpected if we have a small number
of real-time users: There is a good chance that other users may have a
better channel so that once the constraints is satisfied, it is more
useful to allocate the power to these non real-time users since they
will get a better rate.

\begin{figure}
  \centering
  \includegraphics[scale=0.7]{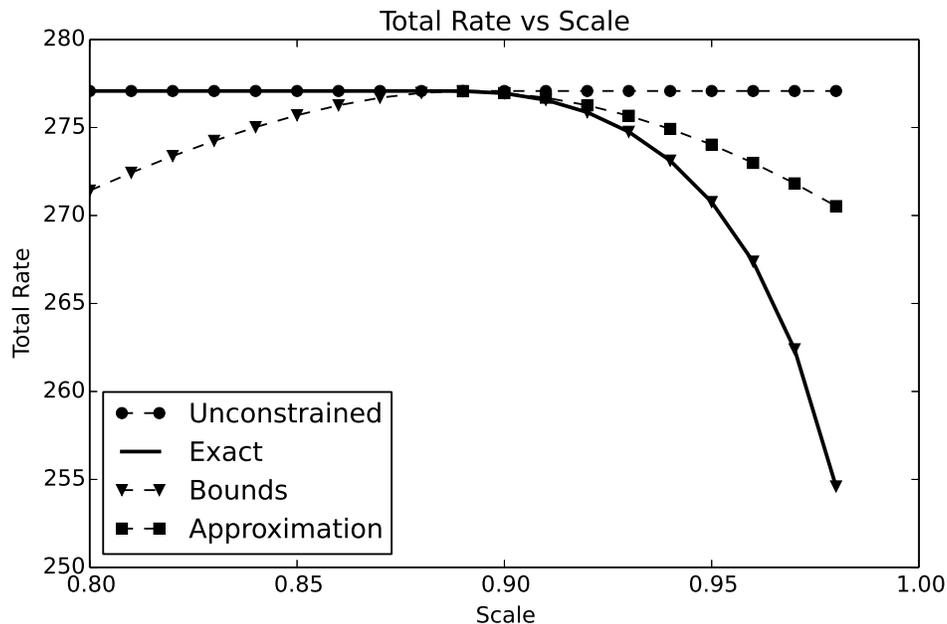}
  \caption{Accuracy of Heuristics, 10 Real-time Users}
  \label{fig:objvsscale-20-P2}
\end{figure}
A similar result is shown for problem~\ref{pr20252510}  on
Fig.~\ref{fig:objvsscale-20-P2}. This has the same parameters as
problem~\ref{pr2025253}  but
this time with 10 real-time users.  Here too we see that the rate
constraints are at their bound in the optimal solution.  The
approximate heuristic does a rather poor job of finding a feasible
solution for the value of $\epsilon = 0.2$ that we have chosen.

The rate bounds in the example of Fig.~\ref{fig:objvsscale-20-P2}
are computed from the problem data and in a sense they ``fit'' the
values of $\beta$. In practice, the rate requirement of real-time
users would be determined by the application independently of the
channel conditions. To see the accuracy of the solution techniques on
these cases, we have produced the bounds somewhat differently. First,
we assume that there is a small number of applications, three in the
present case, with different rate bounds $\check{d}_k^{(0)}$  in the ratio 1:4:16. For
example, this
could correspond to voice, fast data transfer and video. We then
compute the largest scaling factor $\gamma$ such that the problem is
feasible with bounds $\check{d}_k = \gamma \check{d}_k^{(0)}$.  We can then  control how difficult the problem is by
selecting some scale between 0 and 1 to scale the bounds $\check{d}_k$ as above.

\begin{figure}
  \centering
  \includegraphics[scale=0.7]{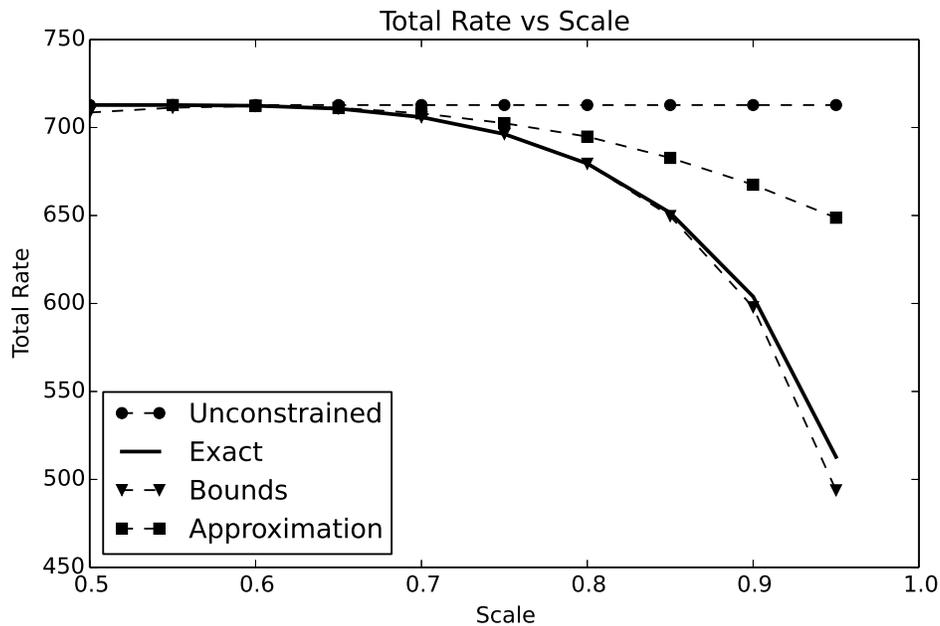}
  \caption{Accuracy for Fixed Rate Bounds, Problem~\ref{pr80258510}}
  \label{fig:accuracy10}
\end{figure}
We show on Fig.~\ref{fig:accuracy10} the accuracy of the algorithm
for problem~\ref{pr80258510} with the given rate bounds. We can see
that the solution on the boundary is still very close to the optimal
solution and the fast heuristic produces an infeasible solution. Note
however that this case has only 10 real-time users out of 80.

\begin{figure}
  \centering
  \includegraphics[scale=0.7]{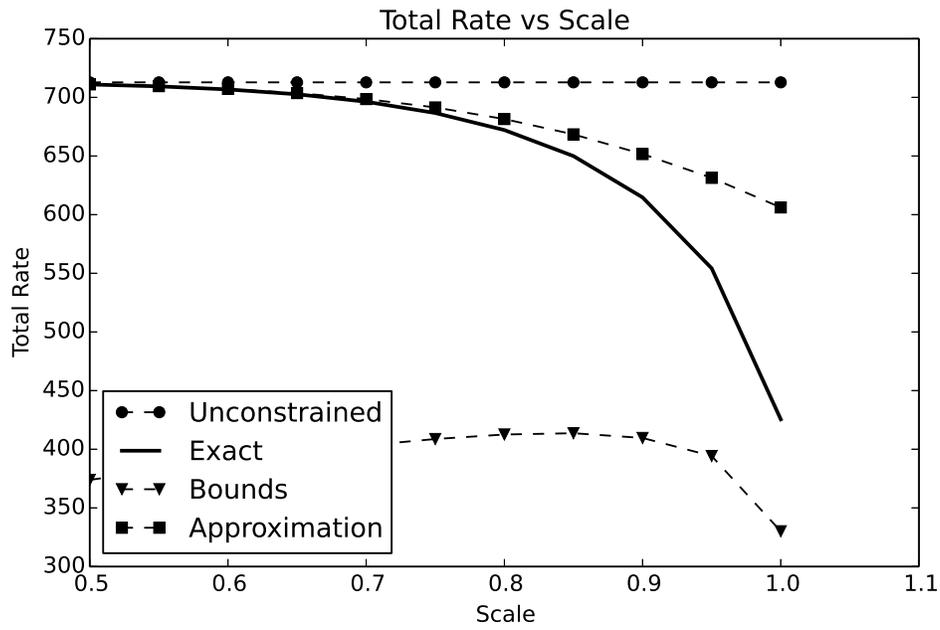}
  \caption{Accuracy for Fixed Rate Bounds, Problem~\ref{pr80258560}}
  \label{fig:accuracy60}
\end{figure}
The situation is quite different if we have 60 real-time users as in
problem~\ref{pr80258560}. These users are divided into 3 groups of 20
and all  users in a group have  the same requirement in the order 1:4:16. We can see from
Fig.~\ref{fig:accuracy60} that the solution on the rate bounds is much lower than the optimal value
for quite a large range of scales and the fast heuristic again produces
infeasible solutions.

\begin{figure}
  \centering
  \includegraphics[scale=0.7]{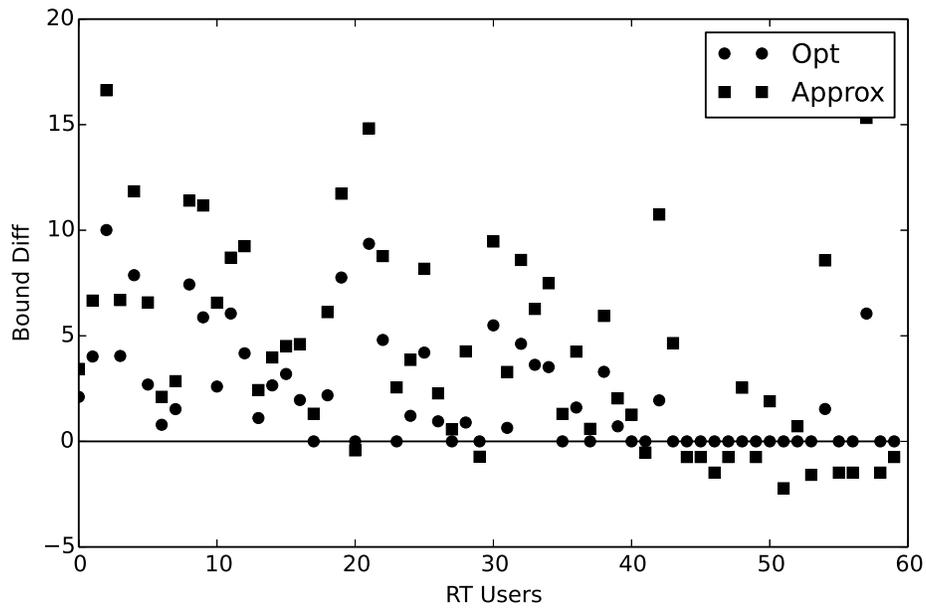}
  \caption{Difference with Rate Bounds}
  \label{fig:diff60}
\end{figure}
We can get a better understanding of these
results from Fig.~\ref{fig:diff60} where we plot, for each real-time
user and  scale 1, the difference $r_k - \check{d}_k$ for the optimal solution and
the approximation. As expected, in the optimal solution, all users are
above zero, since the solution has to be feasible. More important is
the fact that for the first two groups, the ones with the lower
bounds, most users actually get \emph{more}
that their minimum rate. This explains the poor behavior of the bounds
algorithms which computes an optimal solution on the boundary. Also
interesting is the approximate solution for the third group, the one
with high rate. We see that the approximate solution is infeasible but
that the amount of infeasibility is rather small. Note that this
is a relatively hard problem in the sense that the problem is barely
feasible in the first place. Given that the solution produced by the
bounds algorithm is not very good, a reasonable solution would then be
to use the approximation and give some real-time users slightly less
than what they require. This may be acceptable  considering that the channel conditions will
change at the next time frame and that it may then be possible to give
these users the rate that they require.

\subsection{CPU  Time}
\begin{figure}
  \centering
  \includegraphics[scale=0.7]{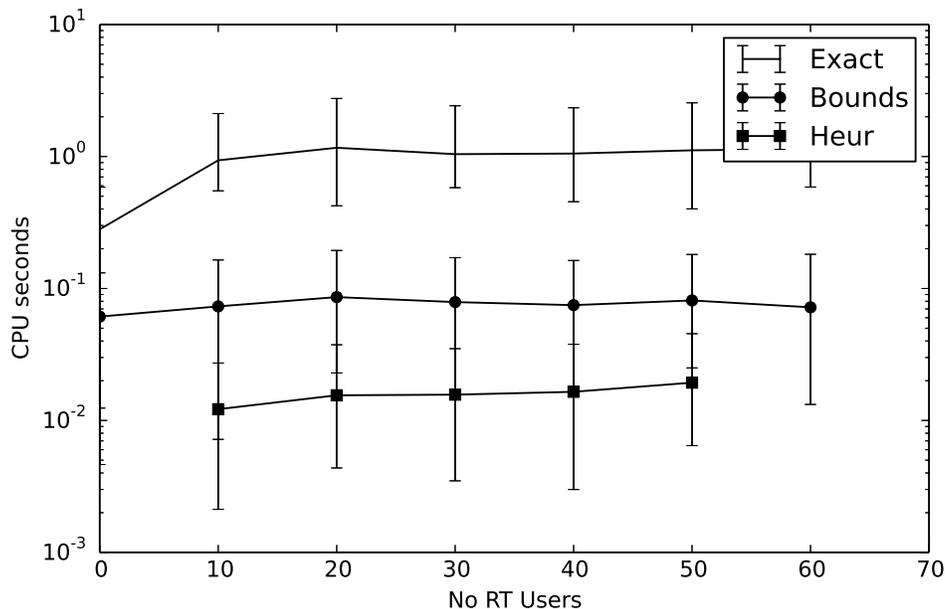}
  \caption{CPU Time vs No of RT Users}
  \label{fig:cpucompare}
\end{figure}
We present in Fig.~\ref{fig:cpucompare} the cpu time required  by the nonlinear solver,
labelled \verb|exact|, the algorithm on the rate boundary, labelled
\verb|bounds| and the heuristic, labelled \verb|heur|. The values of
the computation times are plotted as a function of the number of
real-time users $R$ for problems~\ref{pr8025850}--\ref{pr80258560}. For each value of $R$, we have measured the
average, maximum and minimum cpu time required for solving the problem
over the whole range of values of $s$. The maximum and minimum values
are presented as error bars.  We can see that there is about one order
of magnitude between each algorithm.  The heuristic takes in the order
of 10~ms while the bounds algorithm, about 10 times as much.

%
%

\section{Conclusions} \label{sec:conclusions}
We have looked at the problem of allocating power to a given set of
users in a ZF OFDMA-SDMA system with some real-time users with minimum transmission rate
requirements.  The approach is to solve the problem
in two steps. First, we solve without taking into account the rate
constraints. For this, we provide two fast algorithms. The first one is based on
the solution of the first-order optimality condition via a
zero-finding algorithm. We use the structure of the problem to perform a
search over a finite number of intervals for the one containing the
solution. From this, we get the exact value of the solution so that
the algorithm is guaranteed to find a solution in $O(\log_2 N K)$
steps.

The second algorithm is based on re-writing the optimality conditions
as a fixed-point equation which we can solve by repeated
substitution. This is faster than  the zero-finding method but there
is no guarantee that it will converge. We have found that in practice,
convergence always occurs in a small number of iterations for all the
cases we have tested.

If the solution is feasible, this is the optimal solution for the
constrained problem as well. If some real-time users do not get their
minimum rate, we propose two approximate techniques  to compute a
feasible, or nearly feasible solution by adjusting the dual
multipliers. First, we show that the boundary of the rate constraints
is linear in the $(\theta, \delta)$ plane. We then use this to propose
an algorithm to compute a  solution on the rate boundaries. This is
based on finding a zero of the first-order equations  for the value of
$\theta$ only, since we can use the linearity of the rate boundary to
keep the solution feasible.

The second approximation is not iterative and is a one-step adjustment
of the multipliers. It is parameterized in such a way that we can
emphasize either feasibility, at the cost of having a somewhat less
than optimal value for the rate, or efficiency, where the rate is
relatively large but some real-time users get  somewhat less than
their required rate.

Finally, we present some results for the cpu time needed by three
algorithms: the exact solution by a nonlinear solver, the bounds
algorithm and the approximation. We find that there is roughly one
order of magnitude between each, where the bounds algorithm  is an
order of magnitude faster than the exact solution, and the fast
heuristic is an additional order of magnitude faster.

\bibliographystyle{ieeetran}
\bibliography{journal_paper}
%

\end{document}